\documentclass{article}

\usepackage[utf8]{inputenc}
\usepackage[T1]{fontenc}
\usepackage{amsmath,amssymb,amsthm}
\usepackage{stmaryrd}
\usepackage[table]{xcolor}
\usepackage{mathtools}
\usepackage{hyperref}
\usepackage{cleveref}
\usepackage{booktabs}
\usepackage{enumitem}
\usepackage{listings}
\usepackage{xcolor}
\usepackage[margin=1in]{geometry}
\usepackage{braket}
\usepackage{microtype}
\usepackage{authblk}
\usepackage{blkarray}
\usepackage{thmtools}
\usepackage{tikz-cd}
\usepackage{physics}
\usepackage[most]{tcolorbox}

\newtcolorbox{mybox}[1][]{
    breakable,                
    colback=blue!5!white,      
    colframe=blue!75!black,    
    fonttitle=\bfseries\large,
    title={#1},       
    enhanced,                  
    attach boxed title to top left={yshift=-2mm, xshift=2mm},
    boxed title style={colback=blue!75!black},
    #1                         
}

\DeclareMathOperator{\Fun}{Fun}

\lstdefinelanguage{Lean}{
  morekeywords={def,theorem,lemma,structure,where,abbrev,noncomputable,let,fun,if,then,else,by,sorry,import,variable,instance,section,end,open,set_option,example,class,inductive,namespace,match,with,do,return,have,show,calc,suffices,obtain,rcases,refine,exact,apply,intro,intros,rw,simp,unfold,dsimp,constructor,cases,induction},
  sensitive=true,
  morecomment=[l]{--},
  morecomment=[s]{/-}{-/},
  morestring=[b]",
  literate={:=}{$\coloneqq$}2
}

\newcommand{\Draft}{}

\ifdefined\Draft
\newcommand{\Ethan}[1]{{\color{magenta} [Yi: #1]}}
\newcommand{\RT}[1]{{\color{red} [Runzhou: #1]}}

\else
\newcommand{\Ethan}[1]{}
\newcommand{\RT}[1]{}
\newcommand{ME}[1]{}
\fi

\newcommand{\bbN}{\mathbb{N}}
\newcommand{\bbU}{\mathbb{U}}
\newcommand{\bbP}{\mathbb{P}}
\newcommand{\PauliFold}{\mathsf{PauliFold}}
\DeclareMathOperator{\Image}{Image}

\newtheorem{theorem}{Theorem}[section]
\newtheorem{lemma}[theorem]{Lemma}

\theoremstyle{definition}
\newtheorem{definition}[theorem]{Definition}

\newcommand\doubleplus{+\kern-1.3ex+\kern0.8ex}

\title{\textbf{End-to-End Formalization of Quantum Error Correction}}
\author[ ]{Mattias Ehatamm}
\author[ ]{Yi Lee}
\author[ ]{Xiaodi Wu}
\author[ ]{Runzhou Tao}
\affil[ ]
{Joint Center for Quantum Information and Computer Science, University of Maryland, College Park}
\affil[ ]{\texttt{\{mehatamm, ylee1228, xiaodiwu, rztao\}@umd.edu}}

\date{}

\begin{document}
\maketitle

\begin{abstract}
Quantum error-correcting codes (QECCs) sit between noisy quantum hardware and
reliable computation, so the code parameters used in practice must be
trustworthy. The single number that summarizes a code's strength is its
\emph{distance}, yet certifying a distance lower bound is NP-hard in general,
placing it beyond the reach of pen-and-paper proofs as well as direct
proof-assistant scripting. As a result, distance values in the literature
come either from non-scaling hand proofs, or from
unverified solvers that leave a trust gap exactly where the code is supposed to
provide a guarantee.
We present \textsc{Lean-QEC}, the first Lean 4 formalization of stabilizer-code
theory that delivers end-to-end, machine-checked distance certificates at
industrial code sizes. \textsc{Lean-QEC} formalizes the linear algebra of qubit
states, the Pauli group, stabilizer codes, the binary symplectic representation,
classical coding theory, and the CSS and Bivariate Bicycle families. To break the combinatorial barrier, \textsc{Lean-QEC} translates the distance
condition into a Boolean satisfiability formula through a verified reduction. The pipeline scales through a \textsc{BitVec}-flattened
encoding that replaces Lean's \textsc{Matrix} representation, and an
error-\emph{location} encoding that reduces the variable count from $n$ to
$k\lceil \log_2 n\rceil$. With these, we obtain automatically-generated Lean-checked distance proofs for a large range of industrially viable qLDPC codes within the Bivariate Bicycle and Generalized Bicycle families, including $\llbracket 90, 8, 10 \rrbracket$ and $\llbracket 70, 6, 9 \rrbracket$ BB codes, with the formulation scaling up to 144 qubits when performed outside the Lean kernel. The resulting library is reusable and is designed to
plug into broader Lean-based efforts toward end-to-end verification of
fault-tolerant quantum computation.
\end{abstract}

\section{Introduction}\label{sec:intro}

Quantum computing promises asymptotic speedups on classically intractable
problems in cryptography, simulation, optimization, and
chemistry~\cite{shor1999polynomial, berry2007efficient, mcardle2020quantum,
harrow2009quantum}. Realizing those speedups on real hardware, however, is
constrained by noise: every physical qubit decoheres and every physical gate is
imperfect~\cite{Preskill_2018}, so the only known route to scalable computation
passes through a thick layer of \emph{quantum error correction} (QEC). Recent
experimental demonstrations of QEC on superconducting and trapped-ion
platforms~\cite{google2023suppressing, postler2024demonstration,
google2025quantum} have made this layer a load-bearing component of the quantum
stack rather than a theoretical curiosity, and have accordingly raised the bar
for the trust placed on it.

The classical engineering response to such load-bearing code is testing, but
testing is essentially unavailable for quantum programs. Output distributions
are exponentially expensive to sample, intermediate measurements collapse the
very state under inspection, and the noise that motivates correction also
corrupts any would-be observation. The remaining option is to \emph{formally
verify} the QEC layer~\cite{demoura2021, mathlib2020}: state its guarantees as
theorems and discharge them in a proof assistant whose kernel can be
independently trusted.

\paragraph{Why the distance number is hard to trust.} The most informative
single parameter of a stabilizer code is its \emph{distance} $d$: the minimum
Pauli weight of an undetectable error. A code with distance $d$ corrects
$\lfloor(d-1)/2\rfloor$ errors, so distance is the value most often quoted when
codes are compared and the value on which fault-tolerance thresholds depend. Yet
certifying a distance lower bound is computationally hard: even for classical
linear codes the problem is NP-complete~\cite{vardy2002intractability}, and for
stabilizer codes lacking strong symmetry the situation inherits the same
intractability \cite{Kapshikar_2023}. The values quoted in the literature therefore come from one of
two unsatisfying places:
\begin{itemize}[leftmargin=1.5em]
    \item \textbf{Hand proofs inside a proof assistant.} Existing developments
    such as the 9-qubit Shor-code verification~\cite{feng-9-qubits} enumerate
    every Pauli error explicitly. The bookkeeping --- $3n$ single-qubit error
    patterns plus all of their combinations --- has so far prevented this style
    from scaling past nine qubits, well short of any code used in current
    experiments.
    \item \textbf{Unverified symbolic computation.} ILP solvers, custom search
    routines, or computer-algebra scripts can in principle handle larger codes~\cite{webster2026distancefindingalgorithmsquantumcodes}, 
    but their outputs sit entirely outside the proof kernel. A distance value
    pasted from such a tool is the moral equivalent of a numerical constant
    trusted on faith.
\end{itemize}
Neither option is end-to-end: the first does not scale, and the second leaves a
trust gap at exactly the point the code is meant to provide a guarantee.

\paragraph{Why end-to-end matters.} A pragmatic shortcut is to declare a
stabilizer code by its binary symplectic representation and treat distance as a
property of that matrix. This makes the math tractable, but it shifts the trust
to an informal equivalence between two definitions --- the matrix view and the
original quantum-theoretic view in which a stabilizer code is the simultaneous
$+1$ eigenspace of a set of Pauli matrices. A genuinely trustworthy artifact
must close that loop inside the proof kernel: distance bounds proved over the
matrix view must transport, by a machine-checked theorem, to the eigenspace
view.

\paragraph{Our approach.} We present \textsc{Lean-QEC}, a Lean~4 library that
builds the stabilizer-code stack end-to-end and discharges distance certificates
against industrial-scale codes by combining the proof assistant with an external
SAT solver. The library covers the linear algebra of qubit states, the Pauli
group with phases, stabilizer codes, the binary symplectic representation,
classical coding theory, CSS codes, and the Bivariate Bicycle (BB) family. Each
abstraction layer is connected to the next by a machine-checked correspondence
theorem, so a distance bound established on a binary symplectic matrix is also a
theorem about the corresponding quantum code.

For the computationally hard step, \textsc{Lean-QEC} reduces ``distance $\geq
d$'' to the unsatisfiability of a Boolean formula via a Lean-verified
translation, hands the formula to CaDiCaL through Lean's \lstinline{bv_decide}
tactic, and reconstructs the resulting LRAT certificate as a Lean proof term.
The trusted base remains Lean's kernel plus the LRAT checker; the SAT solver
itself is not trusted.

Making this pipeline scale required two non-obvious design moves, both of which
we present as contributions in their own right:
\begin{enumerate}[leftmargin=1.5em]
    \item Instantiating the Boolean formula directly through Mathlib's
    \lstinline{Matrix} type exhausts memory well before reaching the
    $[[90,8,10]]$ BB code. We introduce a \lstinline{BitVec}-flattened
    intermediate representation, prove it consistent with the
    \lstinline{Matrix}-based reference, and use it to carry the encoding past the 90-qubit mark.
    \item Even with an efficient encoding, a one-Boolean-per-qubit error
    variable produces a formula whose solver runs become intractable beyond
    about 90 qubits. We instead let the variables encode the \emph{locations} of
    the (at most $d-1$) nonzero error entries, shrinking the variable count from
    $n$ to $k\lceil\log_2 n\rceil$. For the $[[144,12,12]]$ BB code this is $88$
    variables versus $144$ --- a difference that makes 144 qubits tractable
    outside the kernel in minutes.
\end{enumerate}

\paragraph{Contributions.} Our contributions are as follows.
\begin{enumerate}[leftmargin=1.5em]
  \item \textbf{An end-to-end Lean~4 formalization of stabilizer-code theory.}
  We give a machine-checked development that connects qubit states, the Pauli
  group with phases, stabilizer codes, the binary symplectic representation, and
  classical coding theory, with explicit correspondence theorems between
  adjacent layers (\Cref{sec:QECBackground}).

  \item \textbf{A verified reduction from code distance to Boolean
  satisfiability.} We formalize a translation from the stabilizer-theoretic
  distance condition to an UNSAT query, prove the reduction sound in Lean, and
  integrate it with Lean's \lstinline{bv_decide} tactic so the SAT certificate
  is reconstructed in the kernel (\Cref{sec:smt}).

  \item \textbf{Two scalability techniques.} A \lstinline{BitVec}-based
  intermediate representation, verified consistent with the \lstinline{Matrix}
  encoding, lets the formula instantiation reach industrial code sizes; a
  location-indexed error encoding cuts the SAT variable count from $n$ to
  $k\lceil\log_2 n\rceil$ (\Cref{sec:translation}).

  \item \textbf{Lean-checked distance certificates for practical codes.} We
  instantiate the pipeline on industrially viable codes like the $[[90,8,10]]$ Bivariate Bicycle code of \cite{bravyi2024} and the $[[70, 6, 9]]$ BB code of \cite{tripier2026faulttolerantquantumcomputingtrapped}, and report the
  resource envelope outside Lean for the Gross Code of size $[[144,12,12]]$. Each case study is between
  $100$ and $200$ lines of Lean script and follows a uniform template
  (\Cref{sec:casestudies}), which may be automatically generated by a Python script given the code description.
\end{enumerate}

\paragraph{Comparison.} \Cref{tab:comparison} summarizes how \textsc{Lean-QEC}
sits relative to the two existing options. Hand-written ITP proofs of distance
close the trust loop but stall well below twenty qubits; unverified solver
pipelines reach industrial codes but never enter the kernel. \textsc{Lean-QEC}
closes the loop \emph{and} reaches industrial codes by paying for the solver
call with a verified reduction and a reconstructed certificate.

\begin{figure}[t]
    \centering
    \vspace{-1.25cm}
    \includegraphics[width=0.9\linewidth]{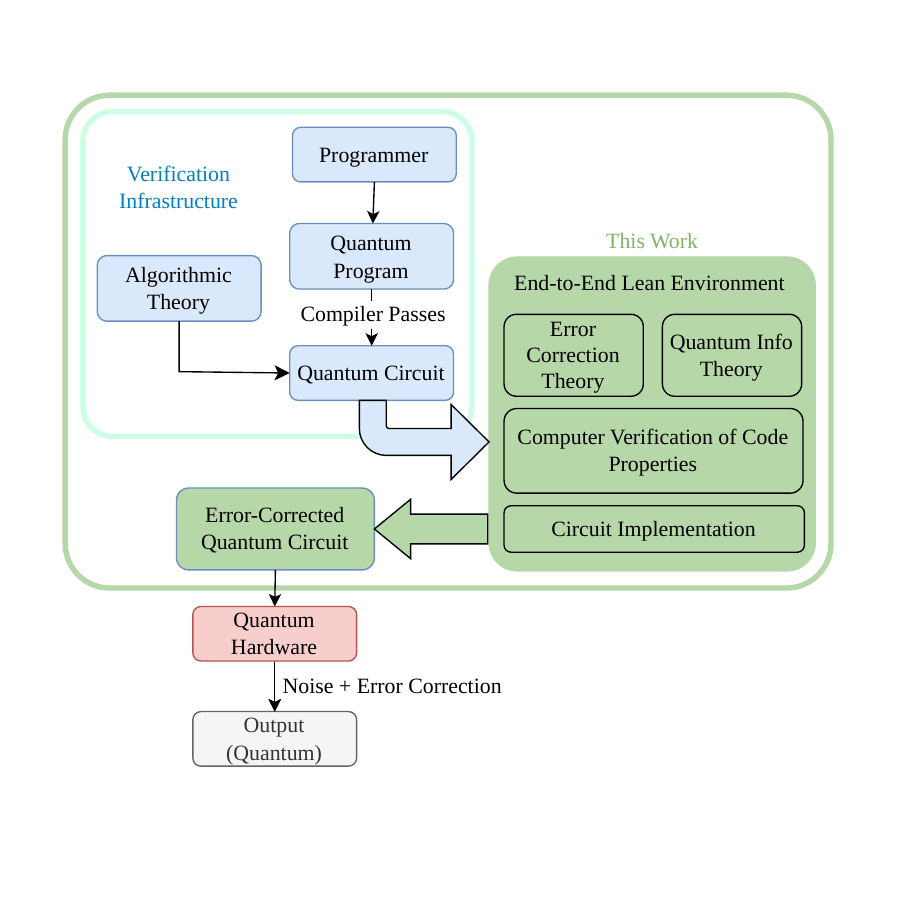}
    \vspace{-2cm}
    \caption{Verification workflow for fault-tolerant quantum programs.
    \textsc{Lean-QEC} occupies the QEC layer of this pipeline.}
    \label{fig:qflow}
\end{figure}

\begin{table}[ht]
\small
\centering
\caption{Comparison of verification approaches for QEC
distance.}\label{tab:comparison}
\begin{tabular}{@{}lccc@{}}
\toprule
\textbf{Approach} & \textbf{Theory verified?} & \textbf{Scales to large codes?} & \textbf{End-to-end checked?} \\
\midrule
Hand-written ITP proofs         & Yes & No ($n > 20$ prohibitive) & Yes \\
Unverified symbolic computation & No  & Yes                       & No  \\
\textsc{Lean-QEC} (this work)   & Yes & Yes (${\sim}2^{40}$ cases) & Yes \\
\bottomrule
\end{tabular}
\end{table}

The rest of the paper develops \textsc{Lean-QEC} in the order a quantum
researcher would encounter the abstractions: \Cref{sec:fv} introduces just
enough Lean to follow the formalization; \Cref{sec:QECBackground} builds the QEC
stack with explicit attention to the type-theoretic gaps that arise when
formalizing standard quantum-information arguments; \Cref{sec:smt} presents the
SAT-assisted distance pipeline; \Cref{sec:casestudies} reports the case studies;
and \Cref{sec:Conclusion and Discussion} discusses related work and limitations.

\section{Background}\label{sec:setting}
\label{sec:fv}

This section gives some background Lean and formal verification for
\Cref{sec:QECBackground,sec:smt}. We assume familiarity with quantum error
correction and standard linear algebra.
For more details please refer to \cite{avigad2020mathematics}.

\paragraph{Lean 4 and Mathlib.} Lean~4~\cite{demoura2021} is an interactive
theorem prover (ITP) and a dependently typed functional programming language
whose underlying logic is the Calculus of Inductive Constructions. Statements
and proofs are both first-class objects: a proof of a proposition $P$ is a value
whose type is $P$. Once Lean has accepted a proof, the only assumptions a reader
needs to discharge are (i)~that the Lean statement faithfully captures the
intended informal claim and (ii)~that Lean's small, independently audited kernel
is correct.

The standard library \emph{Mathlib}~\cite{mathlib2020} now contains close to two
million lines of formalized mathematics and underpins large efforts such as the
Prime Number Theorem~\cite{Kontorovich_Prime_Number_Theorem_2024} and Fermat's
Last Theorem~\cite{FLT_Lean} projects. We rely on Mathlib for matrices,
finite-dimensional vector spaces, group theory, unitary groups, and the
Kronecker product; conversely, our development is designed to be portable back
into Mathlib as quantum-information content matures there.

\paragraph{Formalizing Math with an ITP.} A Lean proof is a sequence of \emph{tactics} that incrementally transforms the goal, and any informal pen-and-paper argument can be translated by spelling out the steps the human reader would otherwise fill in. The friction sits in those omitted steps: edge cases that the prose dismisses with ``clearly,'' implicit coercions between mathematically equivalent representations, and definitions whose set-theoretic phrasing is awkward in type theory.

\begin{mybox}{\textbf{Example: Proving Distance in Lean}}

For a concrete Lean example we revisit a definitional edge case that arises in our own
development. The distance of a quantum error-correcting code is the minimum
Pauli weight of an undetectable error, but a minimum over the empty set is
undefined; we must pick a convention for codes whose only valid codeword is
$0^n$. We follow the standard ``$n+1$'' convention, which makes the distance
well-defined while reserving the literal value $n+1$ as a sentinel for the
trivial case.

\textbf{Definition:} Given an error-correcting code $\mathsf{C}$, let $S\subseteq \mathbb{P}_n$
    be the set of undetectable Pauli errors.
    We then define the distance by
    \begin{equation}
        d = \begin{cases}
            \min_{P\in S}(\mathsf{weight}(P)) & \textnormal{if } S\ne\emptyset \\
            n+1 & \textnormal{if } S = \emptyset
        \end{cases}
    \end{equation}

The sentinel value makes the following sanity check immediate, but only
\emph{after} the case split is made explicit.

\textbf{Lemma:} If the distance $d$ of a code over $n$ qubits satisfies $d \leq n$, then its
    set $S$ of undetectable errors is nonempty.
\begin{proof}
    Suppose $S = \emptyset$ for contradiction. By the definition of code distance, we
    have $d = n+1$, contradicting $d \leq n$.
\end{proof}

A corresponding Lean proof mirroring step-by-step the above informal argument would be written as:

\begin{lstlisting}
lemma BinSympMatrix.undetectable_set_nonempty {k n : ℕ}
    (B : BinSympMatrix k n) (hk : 0 < k)
    (h_dist : B.distance hk ≤ n) :
  B.undetectable_set.Nonempty := by
  by_contra h_empty
  rw [Finset.not_nonempty_iff_eq_empty] at h_empty
  unfold BinSympMatrix.distance at h_dist
  rw [h_empty] at h_dist
  simp at h_dist
\end{lstlisting}    

This proof may be read as follows: \lstinline{by_contra} introduces the
contradiction hypothesis; \lstinline{Finset.not_nonempty_iff_eq_empty} converts
``non-emptiness'' into a Finset-level equation that Mathlib has lemmas about;
\lstinline{unfold} and \lstinline{rw} apply the definition and the empty-set
hypothesis; \lstinline{simp} discharges the residual arithmetic against $d \le
n$. The mechanical overhead, which involves naming the empty-set lemma and finding the right lemmas to rewrite with, is representative of what formalizing a pen-and-paper QEC
argument feels like in practice.

\end{mybox}

\section{Formalization of Quantum Error Correction}\label{sec:QECBackground}

\begin{figure}[!b]
    \centering
    \includegraphics[width=\linewidth]{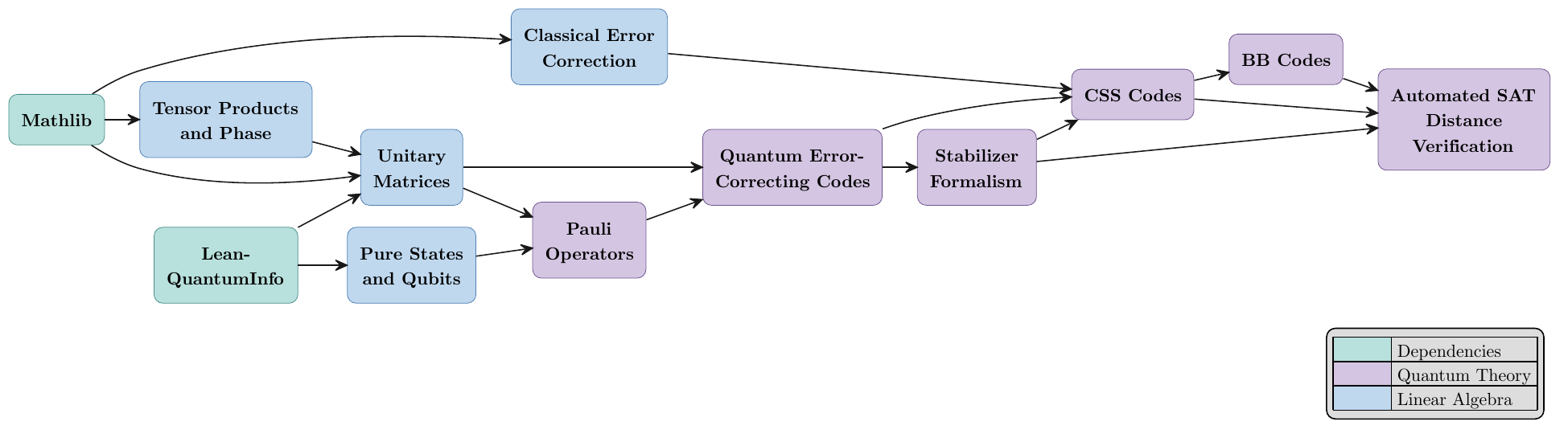}
    \caption{Map of the formalized quantum-error-correction stack. Arrows are
    machine-checked correspondence theorems; the rightmost layer feeds into the
    SAT-based distance pipeline of \Cref{sec:smt}.}
    \label{fig:qmap}
\end{figure}

%

This section builds the stack from \Cref{fig:qmap} bottom-up: qubit states,
$n$-qubit unitaries, the Pauli group with phases, the binary symplectic
representation, abstract QEC codes, stabilizer codes, and finally the CSS and
Bivariate Bicycle families. For each layer we describe the formalization choice,
the machine-checked correspondence to the layer below, and the type-theoretic
friction that distinguishes the Lean development from a pen-and-paper account.
The chain of correspondences ending at the binary symplectic matrix is what
makes the SAT-based distance pipeline of \Cref{sec:smt} end-to-end: a distance
bound discharged at the bottom transports, by a Lean theorem, to a bound on the
quantum code at the top.

\subsection{Qubit States}
\label{sec:qc}

The SQIR development~\cite{hietala2021} and CoqQ~\cite{coqq} formalize quantum
states in Rocq Prover, and the Lean-QuantumInfo project~\cite{leanquantuminfo}
formalizes general finite-dimensional quantum states in Lean. None of these
specialize the state space to qubits in the way we need, so we start the stack
ourselves.

\paragraph{Choice of index set.} A textbook quantum researcher writes a pure
state as
$\ket{\psi} = \sum_{x \in \{0,1\}^n} \alpha_x \ket{x}$,
i.e.\ as an assignment of amplitudes to bit strings. A textbook formalizer
writes it as an element of $\mathbb{C}^{2^n}$. The two views are isomorphic, but
Lean represents a finite-dimensional vector by its function form
\begin{equation}
    \label{eq:vec-equiv-fn}
    \mathbb{F}^n \cong \Fun(\bbN_n, \mathbb{F}), \qquad \bbN_n = \{0, 1, \ldots, n-1\},
\end{equation}
and the choice of index set is a load-bearing design decision: every later
definition is parametrized by it. We diverge from SQIR's $\mathbb{C}^{2^n}$
convention and adopt $\Fun(\{0,1\}^n, \mathbb{C})$, written
$\mathbb{C}^{\{0,1\}^n}$ with the usual normalization. The payoff is that
tensor products on basis states become string concatenation, $\ket{101} \otimes
\ket{001} = \ket{101001}$, instead of an arithmetic identity on bit-shifts; this
matters because Pauli operators are most naturally indexed by which qubit they
act on, not by a position in a flattened $2^n$-array.

\subsection{Quantum Operators}

For the purpose of QEC we do not need a full circuit DSL; we only need to
compose states and operators consistently. The single property we require is
\begin{equation}
\label{eq:kron-ustate}
(U_1 \otimes U_2)(\ket{\phi_1} \otimes \ket{\phi_2}) = (U_1 \ket{\phi_1}) \otimes (U_2 \ket{\phi_2}),
\end{equation}
together with associativity of $\otimes$ on operators. The Lean development for
a richer gate-level circuit calculus is left to future work.

\paragraph{Tensor vs.\ Kronecker.} \Cref{eq:kron-ustate} hides a vocabulary
issue. The product on states is a \emph{tensor product} of vectors; the product
on operators is a \emph{Kronecker product} of matrices. Mathlib distinguishes
them as two different functions, and the abstract isomorphism connecting the two
passes through tensor-algebra machinery we would otherwise have to drag through
every proof. We sidestep this by formalizing every ``$\otimes$'' in our
development as a Kronecker product, with vectors viewed as $n \times 1$ matrices
when needed. This makes \Cref{eq:kron-ustate} a direct corollary of a Mathlib
lemma rather than a tour through Mathlib's tensor-algebra hierarchy.

\paragraph{Index alignment.} Matrices in Mathlib are functions on index pairs,
\begin{equation}
    \label{eq:mat-eq-fn}
    \mathbb{F}^{m \times n} \cong \Fun(\bbN_m \times \bbN_n, \mathbb{F}),
\end{equation}
and Mathlib's Kronecker product is defined accordingly: for $M_1 \in
\mathbb{F}^{I_1 \times J_1}$ and $M_2 \in \mathbb{F}^{I_2 \times J_2}$,
\begin{equation}
    (M_1 \otimes M_2)((i_1, i_2), (j_1, j_2)) = M_1(i_1, j_1) \cdot M_2(i_2, j_2).
\end{equation}
For an $n_1$-qubit operator $U_1$ and an $n_2$-qubit operator $U_2$, the index
type produced by this definition is $\{0,1\}^{n_1} \times \{0,1\}^{n_2}$,
whereas the index type of an $(n_1+n_2)$-qubit operator is $\{0,1\}^{n_1+n_2}$.
The two are equal up to the isomorphism $\{0,1\}^{n_1} \times \{0,1\}^{n_2}
\cong \{0,1\}^{n_1+n_2}$ that concatenates bit strings, but Lean's dependent
type system does not coerce them implicitly. We discharge this gap by a one-time
concatenation/splitting bijection; every later lemma uses operators indexed by
$\{0,1\}^{n}$ uniformly.

\paragraph{Heterogeneous equality.} Associativity of $\otimes$ on operators, $(U_1 \otimes U_2) \otimes U_3 = U_1
    \otimes (U_2 \otimes U_3)$,
is well-formed on paper but not in Lean: the LHS has type indexed by
$(n_1+n_2)+n_3$ qubits and the RHS by $n_1+(n_2+n_3)$, and $\mathbb{N}$'s
addition is not associative \emph{definitionally}. Other formalizations work
around this with explicit cast functions sprinkled through every proof. Lean's
Mathlib provides a cleaner option, \emph{heterogeneous equality}: the equation
is stated between values of provably-equal types, and the type checker accepts
it once the type-level equality has been discharged. We use this idiom
throughout the development.

\paragraph{Unitarity.} Mathlib already supplies the unitary-group structure on
square matrices, and Lean-QuantumInfo uses it as well; we inherit it directly.
Wrapping our Kronecker-based operators in Mathlib's unitary group introduces no
additional conceptual challenges, only boilerplate that we have written once.

\subsection{Pauli Group}

Quantum errors in our setting are decomposed into Pauli operators, so the Pauli
group is the first non-trivial structure in the QEC stack. The standard $X^x
Z^z$ shorthand --- ``apply $X$ to qubits with $x_i = 1$ and $Z$ to qubits with
$z_i = 1$'' --- is convenient for pen-and-paper arguments: the tensor product
becomes string concatenation, $X^{x_1}Z^{z_1} \otimes X^{x_2}Z^{z_2} = X^{x_1 \|
x_2}Z^{z_1 \| z_2}$, and the matrix product becomes qubit-wise. The same
shorthand also underlies the binary symplectic representation of
\Cref{sec:binsymp}.

\paragraph{What the shorthand omits.} Two gaps prevent us from using the $X^x
Z^z$ form directly as a Lean definition.
\begin{itemize}
    \item It describes a \emph{recipe} for building a unitary one qubit at a
    time, not a single matrix in $\mathbb{U}_{\{0,1\}^n}$. The Lean development
    needs the global matrix to state and prove e.g.\ commutation lemmas.
    \item It ignores global phase. The standard Pauli group contains $Y = -iXZ$
    and the anticommutation identity $ZX = -XZ$ depends on the sign; neither is
    expressible without phases. Restricting to the multiplicative subgroup
    $\langle X, Z\rangle$ does not recover $Y$.
\end{itemize}
We close both gaps explicitly. We start from the single-qubit Pauli matrices and
define a \emph{folding} function that turns a phase factor and an $n$-tuple of
single-qubit Paulis into a single $n$-qubit unitary. The Pauli group is then the
image of this function.
\begin{definition}[Single-qubit Pauli gates]
    \label{def:pauli-1}
    We define the following four unitary matrices (acting on the basis
    $\set{\ket{0}, \ket{1}}$) as the single-qubit Pauli gates.
    $$I = \begin{pmatrix}
    1 & 0 \\ 0 & 1
 \end{pmatrix}\quad X = \begin{pmatrix}
    0 & 1 \\ 1 & 0
\end{pmatrix}\quad Y = \begin{pmatrix}
    0 & -i \\ i & 0
\end{pmatrix} \quad Z = \begin{pmatrix}
    1 & 0 \\ 0 & -1
\end{pmatrix}$$
\end{definition}
%

The folding step assembles an $n$-tuple of single-qubit Paulis together with one
of the four phase factors $\{1, i, -1, -i\}$ into a single $n$-qubit unitary.
\begin{definition}[Pauli folding]
    \label{def:pauli-folding}
    For each $n \in \bbN$, define
    $$\PauliFold_n : \{1, i, -1, -i\} \times \{I, X, Y, Z\}^n \rightarrow
    \bbU_{\{0,1\}^n}$$
    by
    \begin{equation}
        \label{eq:def-pauli-fold}
        \PauliFold_n(\gamma, \vec{P}) = \gamma \cdot \bigotimes_i P_i.
    \end{equation}
\end{definition}
\begin{definition}
    \label{def:pauli-group}
    The $n$-qubit \emph{Pauli group} is the image of $\PauliFold_n$:
    \begin{equation}
        \label{eq:def-pauli-group}
        \bbP_n = \Image(\PauliFold_n).
    \end{equation}
\end{definition}
This is equivalent to the textbook definition. A common alternative is to
quotient by global phase and work with equivalence classes; we prefer the
image-of-fold form because it lets us write down a single matrix rather than a
coset representative, which keeps later proofs computable.

\begin{figure}
    \centering
    \begin{tikzcd}[column sep=large, row sep=large]
    P,\, P' \arrow{r} \arrow{d}{\mathsf{unfold}} &
    P \otimes P' \arrow{d}{\mathsf{unfold}} \\
    (\gamma, \{P_i\}_i),\, (\gamma', \{P_i'\}_i) \arrow[r] &
    (\gamma\gamma', \{P_i\}_i \doubleplus \{P_i'\}_i)
    \end{tikzcd}
    \caption{Tensor product of unfolded Pauli operators}
    \label{fig:pauli-commute}
\end{figure}

We verify three properties of $\PauliFold$ that drive the rest of the
development.

\emph{First, $\PauliFold$ is a bijection.} Every Pauli operator therefore has a
unique $\gamma \cdot X^x Z^z$ form. This matters for the distance pipeline of
\Cref{sec:distproblem}: distance is defined as the minimum \emph{Pauli weight}
of an undetectable error, which is only well-defined when each operator has a
canonical decomposition.

\emph{Second, $\PauliFold$ respects products.} The diagram in
\Cref{fig:pauli-commute} commutes for $\otimes$ on operators, and an analogous
diagram commutes for ordinary matrix product once phase factors are collected.
These commutation theorems are what let later proofs argue at the $\gamma \cdot
X^x Z^z$ level even though the underlying definition lives at the matrix level.
As a byproduct we also obtain closure under multiplication.

\emph{Third, commutation reduces to a single bit.} Whether two Paulis
$X^{x_1}Z^{z_1}$ and $X^{x_2}Z^{z_2}$ commute or anticommute is what determines
whether an error is detectable by syndrome measurements
(\Cref{sec:distproblem}); we formalize the local check on the single-qubit case
first and then lift.
\begin{lemma}
    [Characterization of single-qubit Pauli commutation]
    We define the \emph{anti-commutation indicator} function as
    \begin{equation}
        \label{eq:pauli-pauli-phase}
        \phi(P, Q)=\begin{cases}
            0 & \textnormal{if } P=I\wedge Q=I\wedge P=Q \\
            1 & \textnormal{otherwise}
        \end{cases}
    \end{equation}
    For all $P, Q\in\set{I, X, Y, Z}$, we then have
    $$PQ=(-1)^{\phi(P, Q)} QP.$$
\end{lemma}
%

The $n$-qubit case is the XOR of the single-qubit indicators.
\begin{lemma}
    [Characterization of Pauli commutation]
    Let $P, Q \in \bbP_n$, and let $P_i, Q_i \in \{I, X, Y, Z\}$ be the
    components acting on the $i$-th qubit. Then $PQ = (-1)^a QP$ for $a =
    \bigoplus_{i=1}^n \phi(P_i, Q_i)$.
\end{lemma}

%

The commutation check is implemented as a \emph{computable} function: $\phi$ is
defined by a case split, not via an existential quantifier, and the development
avoids the axiom of choice. When the Lean elaborator encounters a concrete pair
of Paulis, $\phi$ reduces to a Boolean value automatically and downstream proofs
simplify without further work. This is the property that allows the SAT
translation of \Cref{sec:smt} to discharge concrete distance queries by
unfolding rather than by a chain of rewrites.

\subsection{Binary Symplectic Representation}
\label{sec:binsymp}

The four single-qubit Pauli matrices are encoded in two bits via the standard
isomorphism
$$I \leftrightarrow (0 \mid 0), \quad X \leftrightarrow (1 \mid 0), \quad Z
\leftrightarrow (0 \mid 1), \quad Y \leftrightarrow (1 \mid 1),$$
under which phaseless Pauli multiplication becomes vector addition over
$\mathbb{F}_2^2$ and commutation is detected by the \emph{symplectic product}
$$(a \mid b) \odot (c \mid d) = ad - bc \in \mathbb{F}_2,$$
with $0$ for commutation and $1$ for anticommutation. Lifting to $n$ qubits
gives a representation in $\mathbb{F}_2^{2n}$ and the inner-product form
$$(a \mid b) \odot (c \mid d) = a \cdot d - b \cdot c \in \mathbb{F}_2,$$
where $a, b, c, d \in \mathbb{F}_2^n$. As a sanity check, $X \otimes X \otimes I
= (110 \mid 000)$ and $Z \otimes Y \otimes I = (010 \mid 110)$ have symplectic
product $0$, so they commute despite the two qubits on which they anticommute
pairwise.

\paragraph{Two weight functions.} The Binary Symplectic representation supports
two distinct weight notions, and confusing them silently changes which distance
is being computed. The \emph{Pauli weight} of a vector $(x \mid z)$ counts the
qubits at which $(x_i, z_i) \neq (0, 0)$, which is equivalently the number of qubits
at which the encoded operator is not $I$. The \emph{binary weight} is the
Hamming weight of $x$ plus the Hamming weight of $z$ and over-counts the $Y$
positions. Distance arguments require the Pauli weight; we expose both functions
and prove the relationship between them.

\begin{lstlisting}
abbrev BinSympPauli (n : ℕ) := (Fin n → (ZMod 2)) × (Fin n → (ZMod 2))

def union_weight (v₁ v₂ : Fin n → (ZMod 2)) : ℕ := hammingNorm (fun i => ZMod2_or (v₁ i) (v₂ i))

def BinSympPauli.weight {n : ℕ} (bsp : BinSympPauli n) : ℕ := hammingNorm bsp.1 + hammingNorm bsp.2

def symplecticProd {n : ℕ} (bp₁ bp₂ : BinSympPauli n) : (ZMod 2) :=
  dotProduct bp₁.1 bp₂.2 + dotProduct bp₁.2 bp₂.1
\end{lstlisting}

%

The correspondence between the Pauli group and its Binary Symplectic
representation is built on top of a single-qubit equivalence, which we then lift
qubit-wise. The correspondence is exact in one direction --- every binary
symplectic vector decodes to a definite phaseless Pauli --- but loses phase in
the other; we therefore expose a backwards correctness statement that is itself
quantified over an unknown phase.

\begin{lstlisting}
def Z2Z2_Pauli_equiv : ((ZMod 2) × (ZMod 2)) ≃ Pauli

def BinSympPauli_toPauli {n : ℕ} (bs : BinSympPauli n) : PauliGroup_group n :=
  let pm := fun n => Z2Z2_Pauli_equiv (bs.1 n, bs.2 n)
  foldPauli (pgphase_1, pm)

def Pauli_toBinSymp {n : ℕ} (P : PauliGroup_group n) : BinSympPauli n :=
  let pm := fun n => Z2Z2_Pauli_equiv.symm (P.map n)
  ⟨fun n => (pm n).1, fun n => (pm n).2⟩

theorem Pauli_to_BinSymp_correct {n : ℕ} (P₁ P₂ : PauliGroup_group n) :
  Pauli_toBinSympPauli (P₁ * P₂) = (Pauli_toBinSympPauli P₁) + (Pauli_toBinSympPauli P₂) := sorry

theorem BinSympPauli_to_Pauli_correct {n : ℕ} {bp₁ bp₂ : BinSympPauli n} :
  ∃ p : (pgroup_phases), U_phase (BinSympPauli_toPauli (bp₁ + bp₂)).1 p = ((BinSympPauli_toPauli bp₁) : Uₙ[n]) * (BinSympPauli_toPauli bp₂)

theorem commutes_of_sympProd_zero {n : ℕ} {bp₁ bp₂ : BinSympPauli n}
  (h_prod : symplecticProd bp₁ bp₂ = 0) : commute (BinSympPauli_toPauli bp₁) (BinSympPauli_toPauli bp₂)
\end{lstlisting}

%

The signature of \lstinline{BinSympPauli_to_Pauli_correct} reflects this phase
loss: the equality is quantified over a phase \lstinline{p} $\in \{\pm 1, \pm
i\}$ such that the two unitaries agree after \lstinline{p} is applied.
Commutation, in contrast, is phase-insensitive, so
\lstinline{commutes_of_sympProd_zero} states the implication cleanly. This is
the crucial property we use later: distance arguments care only about
commutation with stabilizers, so the loss of phase in the symplectic encoding
does not cost us any soundness.

\subsection{Quantum Error Correction}\label{sec:qec}

%
%

A quantum error-correcting code can be defined at several levels of generality.
We expose two and use whichever is convenient at each layer of the development.

\begin{definition}[Quantum Error-Correcting Code, \cite{gottesman2024surviving}]
    A \textbf{Quantum Error-Correcting Code} is a pair $(U, \mathcal{E})$ where
    $U$ is a partial isometry $\mathbb{C}^{2^k} \rightarrow \mathbb{C}^{2^n}$
    (an inner-product preserving map) and $\mathcal{E}$ is a set of linear maps
    $\mathbb{C}^{2^n} \rightarrow \mathbb{C}^{2^m}$. $U$ is the encoding
    function, while $\mathcal{E}$ is the set of \textbf{correctable errors}.
    This is a valid QECC if:

    $$\exists (\mathcal{D} : \mathbb{C}^{2^m} \rightarrow \mathbb{C}^{2^k}),
    \forall E \in \mathcal{E}, D(EU\ket{\psi}\bra{\psi}U^\dagger E^\dagger) =
    c(E, \ket{\psi})\ket{\psi}\bra{\psi}$$
\end{definition}

%
%

Read informally: a decoder exists that returns the encoded state to its original
form modulo an error-and-state-dependent scalar. The definition is general but
inconvenient in a formalized setting because it bakes in both an encoder and a
list of correctable errors. In Lean, we prefer to keep the code space, the
encoder, and the error set as separate objects and bundle them only when needed.
The following equivalent characterization makes this separation natural.

\begin{definition}
    Given $(Q, \mathcal{E})$, where $Q$ is a subspace of $\mathbb{C}^{2^n}$, and
    $\mathcal{E}$ is a set of maps $\mathbb{C}^{2^n}\rightarrow
    \mathbb{C}^{2^m}$, $(Q, \mathcal{E})$ is a \textbf{QECC} iff $\forall
    \ket{\psi},\ket{\phi} \in Q, \forall E_a, E_b \in \mathcal{E}$:

    $$\bra{\psi}E_a^\dagger E_b\ket{\phi}=C_{ab}\braket{\psi|\phi}$$
\end{definition}

%

Pen-and-paper QEC drops the encoder when the logical labeling is irrelevant and
works with the code space alone. The Lean development pays a real cost for
keeping an encoder around: every lemma must thread it as an explicit argument
even when nothing in the lemma uses it. We therefore strip both the encoder and
the error set out of the core type and re-attach them only when needed. The
result is a code space alone, which cannot be a \lstinline{Submodule} of
$\mathbb{C}^{\{0,1\}^n}$ because our \lstinline{PState} type already carries the
normalization condition; we use \lstinline{Set (PState n)} instead.

\begin{lstlisting}
abbrev QCodeSpace (n : ℕ) := Set (PState n)
\end{lstlisting}

\subsection{Stabilizer Codes}\label{sec:stabbackground}

%

Every code we care about in practice --- and every code we verify in
\Cref{sec:casestudies} --- is a stabilizer code, the specialization of the
subspace definition where the code space is determined by a commuting subgroup
of the Pauli group. The rest of this subsection builds the stabilizer-code
abstraction on top of the Pauli-group infrastructure of the previous
subsections.

\begin{definition}
    Given a subspace $Q$ of $\mathbb{C}^{2^n}$, the stabilizer $S(Q)$ is:

    $$S(Q) = \{A \in P_n : \forall \ket{\psi} \in Q, A\ket{\psi}=\ket{\psi}$$
\end{definition}

%
%

$S(Q)$ is a subgroup: closure follows from $UV\ket{\psi} = U\ket{\psi} =
\ket{\psi}$, inverses are inherited from the Pauli group, the subgroup is
abelian, and $-I \notin S(Q)$ for any nontrivial $Q$. The map runs the other way
as well: an abelian subgroup of the Pauli group that does not contain $-I$
determines a code space.

\begin{definition}
    The \textbf{Codespace} corresponding to stabilizer subgroup $S$ is:

    $$T(S)=\{\ket{\psi}: \forall A \in S, A\ket{\psi}=\ket{\psi}\}$$
\end{definition}

%

Phaseless Pauli operators have eigenvalues $\pm 1$, so $T(S)$ is exactly the
simultaneous $+1$ eigenspace of $S$.

\begin{definition}
    A stabilizer code is a QECC $(Q, \mathcal{E})$ where $T(S(Q))=Q$.
\end{definition}

%

Stabilization and its closure properties follow directly from the Pauli-group
infrastructure of \Cref{sec:binsymp}.

\begin{lstlisting}
def stabilizes (U : Uₙ[n]) (ψ : PState n) := ψ.apply U = ψ

theorem stab_prod {U₁ : Uₙ[n₁]} {U₂ : Uₙ[n₂]} {ψ₁ : PState n₁} {ψ₂ : PState n₂}
  (hs₁ : stabilizes U₁ ψ₁) (hs₂ : stabilizes U₂ ψ₂) :
  stabilizes (U₁ ⊗ₙ U₂) (ψ₁ ⊗ₚ ψ₂)

theorem inv_stab {U : Uₙ[n]} {ψ : PState n} (hstab : stabilizes U ψ) : stabilizes U⁻¹ ψ

theorem mul_stab {n : ℕ} {U₁ U₂ : Uₙ[n]} {ψ : PState n} (hstab₁ : stabilizes U₁ ψ) (hstab₂ : stabilizes U₂ ψ) : stabilizes (U₁ * U₂) ψ
\end{lstlisting}

%

On top of the abstract QECC definition, a stabilizer code bundles together (i)~a
code space, (ii)~a commuting subgroup of the Pauli group, and (iii)~a proof that
the subgroup stabilizes the code space. Carrying the commutativity proof inside
the structure makes downstream proofs shorter, avoiding re-deriving it at every use site.

\begin{lstlisting}
structure StabCode (n : ℕ) where
  space : QCodeSpace n
  stabs : Subgroup (@PauliGroup_group n)
  h_comm : IsMulCommutative stabs
  hstab : ∀ ψ ∈ space, ∀ S ∈ stabs, stabilizes S ψ
\end{lstlisting}

%

The condition $-I \notin \mathtt{stabs}$ is derivable from this definition but
we leave it implicit. Code constructions almost always specify a stabilizer
group by its generators rather than as a fully expanded subgroup, so we provide
two structures for the generators: \lstinline{StabSet} for an arbitrary
commuting set of generators, and \lstinline{NormalizedStabSet} for the common
case where all generators have phase $+1$.

\begin{lstlisting}
structure StabSet (n : ℕ) where
  stabs : Finset (PauliGroup_group n)
  comm : ∀ s₁ ∈ stabs, ∀ s₂ ∈ stabs, commute s₁ s₂

structure NormalizedStabSet (n : ℕ) extends StabSet n where
  norm: ∀ s ∈ stabs, PauliGroup.phase s = pgphase_1

def StabCode_of_StabSet {n : ℕ} (S : StabSet n)
 : StabCode n where
   space := {ψ | ∀ s ∈ S.stabs, stabilizes s ψ}
   stabs := Subgroup.closure (SetLike.coe S.stabs)
   h_comm := ...
   hstab := ...
\end{lstlisting}

%
%

A stabilizer code corrects errors by projectively measuring the generators of
its stabilizer group; the eigenvalues observed form the \emph{syndrome}. A
simple-minded definition of ``detectable'' would require every correctable error
to have a unique syndrome, but the Shor $[[9,1,3]]$ code shows this is too
strong: as \Cref{tab:shor_stabilizers} indicates, $Z$ on qubit 1 and $Z$ on
qubit 2 produce the same syndrome (both flip only $M_7$), yet they correct
identically because their product lies in the stabilizer group. The right notion
of ``undetectable'' must therefore exclude errors that act trivially on the code
space, not merely those with non-unique syndromes.

\begin{table}[ht]
\centering
\caption{Stabilizer Generators for the $[[9, 1, 3]]$ Shor Code}
\label{tab:shor_stabilizers}
\begin{tabular}{@{}ccc@{}}
\toprule
\textbf{Generator} & \textbf{Type} & \textbf{Operator String} ($Q_1 \dots Q_9$) \\ \midrule
$M_1$ & Phase-Flip & $Z Z I \quad I I I \quad I I I$ \\
$M_2$ & Phase-Flip & $I Z Z \quad I I I \quad I I I$ \\
$M_3$ & Phase-Flip & $I I I \quad Z Z I \quad I I I$ \\
$M_4$ & Phase-Flip & $I I I \quad I Z Z \quad I I I$ \\
$M_5$ & Phase-Flip & $I I I \quad I I I \quad Z Z I$ \\
$M_6$ & Phase-Flip & $I I I \quad I I I \quad I Z Z$ \\ \addlinespace
$M_7$ & Bit-Flip   & $X X X \quad X X X \quad I I I$ \\
$M_8$ & Bit-Flip   & $I I I \quad X X X \quad X X X$ \\ \bottomrule
\end{tabular}
\end{table}

%
%
%
%
%
%
%
%

\begin{definition}
    The \emph{normalizer} of a stabilizer $S$ is $N(S) = \{A \in P_n : \forall B
    \in S,\ AB = BA\}$.
\end{definition}
A Pauli error that commutes with every stabilizer is invisible to syndrome
measurement and maps the code space to itself. The stabilizer group is abelian,
so its own elements always commute with all stabilizers; these elements act
trivially on the code space and must be excluded. The remaining elements are the
normalizer minus stabilizer, which are exactly the errors that escape detection
while still changing the encoded state.
\begin{definition}
    The \emph{undetectable set} of a stabilizer code with stabilizer $S$ is
    $N(S) \setminus S$. The \emph{distance} is the minimum Pauli weight over the
    undetectable set. A code of distance $d$ \emph{corrects} $\lfloor (d-1)/2
    \rfloor$ errors.
\end{definition}
These map directly into Lean.

\begin{lstlisting}
def StabCode.normalizer {n : ℕ} (SC : StabCode n) : Finset (PauliGroup n) := {N | ∀ s ∈ SC.stabs, commute N s}

def StabCode.undetectable {n : ℕ} (SC : StabCode n) (E : PauliGroup n) : Prop := E ∈ SC.normalizer ∧ E ∉ SC.stabs ∧ PauliGroup.normalized E

def StabCode.undetectable_set {n : ℕ} (SC : StabCode n) : Finset (PauliGroup n) := {E | SC.undetectable E}

def StabCode.distance {n : ℕ} (SC : StabCode n) (hnt : SC.nontrivial) : ℕ := Finset.min' (Finset.image (fun E => pauli_weight E) SC.undetectable_set)
 (Finset.image_nonempty.2 (StabCode.undetectable_set_nonempty SC hnt))
\end{lstlisting}

%
%

Two details in the Lean definition deserve comment. The argument \lstinline{hnt}
to \lstinline{StabCode.distance} establishes that the undetectable set is
nonempty, so the minimum is well-defined (cf.\ the sentinel discussion of
\Cref{sec:fv}). And \lstinline{StabCode.undetectable} restricts attention to
phase-$+1$ Paulis; without this, $-I$ would appear as an undetectable element of
every code --- it commutes with all stabilizers, is not itself a stabilizer, yet
acts trivially on the code space. Because commutation is phase-insensitive, this
restriction loses no genuine errors.

\paragraph{Binary symplectic matrix.} Applying the binary symplectic
representation to a minimal generating set of a stabilizer code yields a
\emph{binary symplectic matrix} (BSM): two $\mathbb{F}_2$ matrices, one for the
$X$-component and one for the $Z$-component, with rows that are linearly
independent and have pairwise symplectic product zero. An undetectable error in
this representation is a vector with symplectic product zero against every row
that does not itself lie in the row space.

\begin{lstlisting}
structure BinSympMatrix (k n : ℕ) where
  X : Matrix (Fin k) (Fin n) (ZMod 2)
  Z : Matrix (Fin k) (Fin n) (ZMod 2)

def BinSympMatrix.toStabSet {n k : ℕ} (bsm : BinSympMatrix k n) (h_comm : bsm.isCommuting) : NormalizedStabSet n where
  stabs := Finset.image (fun i => BinSympPauli_toPauli (bsm.row i)) Finset.univ
  comm := --
  norm := --

def BinSympMatrix.undetectable {k n : ℕ} (B : BinSympMatrix k n) (bsp : BinSympPauli n) : Prop :=
  (∀ i, symplecticProd (B.row i) bsp = 0) ∧ bsp ∉ B.rowSpace

def BinSympMatrix.distance {k n : ℕ} (B : BinSympMatrix k n) (hk : 0 < k) : ℕ := Finset.min' (B.undetectable_set.image BinSympPauli.weight)
(Finset.image_nonempty.2 (B.undetectable_set_nonempty hk))
\end{lstlisting}

%

The end-to-end story rests on the following correspondence theorem, which is the
bridge that transports SAT-discharged distance bounds (at the BSM level) into
theorems about the quantum stabilizer code.

\begin{lstlisting}
theorem BinSympMatrix.distance_eq_distance {k n : ℕ} (B : BinSympMatrix k n) (h_comm : B.isCommuting) (hk : 0 < k) :
  (StabCode_of_BinSympMatrix B h_comm).distance (B.toStabCodeNontrivial_of_k_pos h_comm hk) =
  B.distance hk
\end{lstlisting}

%

The BSR loses phase information. For distance arguments this loss is harmless. An error is an error regardless of phase, and commutation against
stabilizers is phase-insensitive, but other quantum-theoretic claims (e.g.\
relations like $XZ = iY$) become false at the BSR level. In Lean this means
correspondence theorems are stated carefully and we never accidentally lift a
phase-dependent fact from one representation to the other.

\subsection{CSS and BB Codes}\label{sec:cssbb}

%

CSS codes are stabilizer codes whose BSM splits cleanly into an $X$-block and a
$Z$-block, each derived from a classical linear code. We follow the standard
recipe: given classical linear codes $C_1, C_2 \subseteq \mathbb{Z}_2^n$ with
$C_1^\perp \subseteq C_2$, we build the BSM from their parity-check matrices.
Recall that the dual code $C^\perp$ is the orthogonal subspace under the dot
product, a generator matrix has linearly independent rows that span $C$, and a
parity-check matrix is a generator matrix for $C^\perp$.

\begin{definition}
    Given classical linear codes $C_1, C_2 \subseteq \mathbb{Z}_2^n$, where
    $C_1^\perp \subseteq C_2$, a \textbf{CSS Code} is defined starting from the
    following binary symplectic matrix:
$$
\begin{blockarray}{cc}
X & Z \\  
\begin{block}{(cc)}
  0 & H_1 \\
  H_2 & 0 \\
\end{block}
\end{blockarray}
$$

Where $H_1, H_2$ are parity-check matrices for the two codes. 
\end{definition}

%
%
%

Linear independence of the BSM rows is immediate: each block inherits
independence from its parity-check matrix and the two blocks have disjoint
$X$/$Z$ support. The mutual symplectic-product-zero condition splits into three
cases. Two rows from the same block produce a dot product of a nonzero half
against a zero half, which is zero. A row from the $H_1$ block paired with a row
from the $H_2$ block produces a sum of two dot products, one of which is $0
\cdot \cdot = 0$ and the other of which is between an element of $C_1^\perp$ and
an element of $C_2^\perp$; by $C_1^\perp \subseteq C_2$ this is also zero.

\paragraph{Classical coding theory in Lean.} Mathlib does not yet contain a
classical coding theory library, so we built one. The core type is a subspace of
$\mathbb{F}_q^n$; we go from a finite generating set to a matrix via Mathlib's
\lstinline{Submodule.span} rather than the other direction, because
\lstinline{Submodule.span} is the side Mathlib's inner-product/dual machinery is
wired up against. The dual code is then defined through Mathlib's bilinear-form
orthogonal-complement machinery rather than re-derived from scratch.

\begin{lstlisting}
abbrev CodeSpace (n : ℕ) := Submodule (ZMod 2) (Fin n → ZMod 2)

abbrev CodeSpace (n : ℕ) := Submodule (ZMod 2) (Fin n → ZMod 2)

def codeFromGenerators (G : Finset (Fin n → ZMod 2)) : CodeSpace n :=
Submodule.span (ZMod 2) G

def Matrix.toCodeSpace {α : Type*} [Fintype α] (M : Matrix α (Fin n) (ZMod 2)) : CodeSpace n :=
  Submodule.span (ZMod 2) (Finset.univ.image (fun i => M i))

def CodeSpace.dualCode (C : CodeSpace n) : CodeSpace n := (C.orthogonalBilin (dotProductBilin _ _))
\end{lstlisting}

%

Generator and parity-check matrices are then characterized by predicates:
row-wise linear independence plus row-span equal to the code (resp.\ its dual).

\begin{lstlisting}
def CodeSpace.generatorMatrixOf (C : CodeSpace n) (G : Matrix α (Fin n) (ZMod 2))
  : Prop := (LinearIndependent (ZMod 2) G.row) ∧ Submodule.span (ZMod 2) (Finset.image (fun r => G r) Finset.univ) = C

def CodeSpace.parityCheckMatrixOf (C : CodeSpace n) (G : Matrix α (Fin n) (ZMod 2))
  : Prop := C.dualCode.generatorMatrixOf G
\end{lstlisting}

%

Two theorems anchor the classical layer: duality is an involution,
$(C^\perp)^\perp = C$, and dimensions add up, $\dim C + \dim C^\perp = n$. The
latter is what \Cref{lma41} uses to certify an externally supplied parity-check
kernel.

\begin{lstlisting}
theorem dual_dual_eq {C : CodeSpace n} : (C.dualCode).dualCode = C

lemma dual_finrank_eq {C : CodeSpace n} : Module.finrank (ZMod 2) C.dualCode = n - (Module.finrank (ZMod 2) C)
\end{lstlisting}

%

A CSS code is then a stabilizer code whose BSM is assembled from two
parity-check matrices satisfying the dual-containment condition.

\begin{lstlisting}
def CSS_BSM (H₁ : Matrix (Fin k₁) (Fin n) (ZMod 2))
  (H₂ : Matrix (Fin k₂) (Fin n) (ZMod 2)) : BinSympMatrix (k₁ + k₂) n where
    X := Fin.append (fun (_ : Fin k₁) => (fun _ => 0)) H₂
    Z := Fin.append H₁ (fun (_ : Fin k₂) => (fun _ => 0))

theorem CSS_BSM_is_comm {C₁ C₂ : CodeSpace n} (H_dual : (C₁.dualCode : Set (Fin n → ZMod 2)) ⊆ C₂) {H₁ : Matrix (Fin k₁) (Fin n) (ZMod 2)}
  {H₂ : Matrix (Fin k₂) (Fin n) (ZMod 2)} (hpc₁ : C₁.parity_checked_by H₁)
  (hpc₂ : C₂.parity_checked_by H₂) : (CSS_BSM H₁ H₂).isCommuting
\end{lstlisting}

%
%
%
%
%

We take \lstinline{parity_checked_by} (a row-span condition) rather than the
strict \lstinline{parityCheckMatrixOf} predicate (row-span plus
row-independence) for the parity-check arguments. The relaxation matters in
practice: many code constructions in the literature, including the $[[90, 8,
10]]$ BB code we verify in \Cref{sec:casestudies}, specify dependent generator
sets that fail strict independence. The independence can always be re-proved
separately, but baking the relaxation in keeps the CSS-code constructor modular.

\paragraph{Bivariate Bicycle codes.} The CSS family we ultimately scale to is
the Bivariate Bicycle (BB) family --- low-density parity-check codes with high
rate and good distance, currently the leading candidate for fault-tolerant
computation on superconducting hardware. Let $S_n$ be the $n \times n$ cyclic
shift permutation matrix and define
$$x = S_\ell \otimes I_m, \qquad y = I_\ell \otimes S_m.$$
A BB code is built from two trivariate polynomials in $x$ and $y$:
\begin{definition}
    A \emph{Bivariate Bicycle} code is the CSS code with parity-check matrices
    $H_X = [A \mid B]$ and $H_Z = [B^T \mid A^T]$, where $A = A_1 + A_2 + A_3$
    and $B = B_1 + B_2 + B_3$, each $A_i, B_i$ a monomial in $x$ and $y$.
\end{definition}
Our case studies in \Cref{sec:casestudies} verify the $[[90, 8, 10]]$ instance
inside Lean and show that the same encoding scales to $[[144, 12, 12]]$ outside
the kernel.

\section{Solver-Assisted Distance Verification}
\label{sec:smt}

%
%

This section presents the verified distance pipeline at the heart of
\textsc{Lean-QEC}. The goal is to discharge a statement of the form ``the code
defined by these parity-check matrices has distance at least $d$'' inside Lean,
even when $d$ and the code size place the underlying combinatorial check out of
reach of any pen-and-paper or direct ITP argument.

\paragraph{Why SAT, not ILP.} The textbook formulation of distance is an integer
program: minimize the Pauli weight of an error subject to commutation with all
stabilizers and exclusion from the row space. An ILP front-end would be the most
direct match. Unfortunately Lean has no verified ILP backend at present, so we
instead reduce to Boolean satisfiability, prove the reduction sound in Lean, and
use the existing SAT integration.

\paragraph{Trusted base.} Lean's \lstinline{bv_decide}
tactic~\cite{biere2024cadical} bitblasts a goal phrased in \lstinline{BitVec}
and \lstinline{Bool} into CNF, sends the formula to CaDiCaL, receives an
LRAT~\cite{lrat} certificate, and reconstructs the certificate inside Lean's
kernel. The solver itself does not enter the trusted base; only the LRAT checker
and Lean's kernel do. The LRAT certificate is cached on disk, so future builds
skip the SAT call and re-check the proof directly.

\subsection{Reduction to Boolean Satisfiability}
\label{sec:distproblem}
\label{sec:translation}

The stabilizer-code distance is
$$d = \min_{E \in P_n} \mathrm{wt}(E)\ \text{s.t.}\ \forall S \in \text{Stab},\
SE = ES \ \text{and}\ E \notin \text{Stab}.$$
The end-to-end correspondence theorems of \Cref{sec:stabbackground} let us
discharge this query at the BSR level: writing $R(B)$ for the rows of the binary
symplectic matrix,
$$d = \min_{E \in \mathbb{Z}_2^{2n}} \mathrm{wt}(E_x \cup E_z)\ \text{s.t.}\
\forall S \in R(B),\ E \odot S = 0 \ \text{and}\ E \notin
\mathrm{rowspace}(B).$$
Everything in this query is a Boolean predicate except for ``$E \notin
\mathrm{rowspace}(B)$.'' We replace it with the dual formulation: row space
equals the orthogonal complement of the kernel, so $E \notin
\mathrm{rowspace}(B)$ holds exactly when $E \cdot g = 1$ for some generator $g$
of $\ker B$. With a finite generating set $G$ for the kernel, the constraint
becomes a finite OR of $\mathbb{F}_2$-dot-product equations.

\paragraph{CSS specialization.} For CSS codes the search splits exponentially.
It suffices to consider purely-$X$ errors (against $Z$-stabilizers) and
purely-$Z$ errors (against $X$-stabilizers); the code distance is then the
minimum of the two,
$$d = \min(d_X, d_Z),$$
where $d_X$ is the minimum weight of an $X$-only error that commutes with all
$Z$-stabilizers and is not generated by the $X$-stabilizers, and symmetrically
for $d_Z$. Each side reduces to a SAT instance over $\mathbb{F}_2$ variables
with a Hamming-weight bound, exactly the shape \lstinline{bv_decide} accepts.

%
\begin{mybox}{\textbf{Worked example: Shor 9-qubit code.}}
We trace the reduction on a small
CSS code with parity-check matrices

$$H_Z = \begin{pmatrix}
    1 & 1 & 0 & 0 & 0 & 0 & 0 & 0 & 0  \\
    0 & 1 & 1 & 0 & 0 & 0 & 0 & 0 & 0 \\
    0 & 0 & 0 & 1 & 1 & 0 & 0 & 0 & 0 \\
    0 & 0 & 0 & 0 & 1 & 1 & 0 & 0 & 0\\
    0 & 0 & 0 & 0 & 0 & 0 & 1 & 1 & 0\\
    0 & 0 & 0 & 0 & 0 & 0 & 0 & 1 & 1\\
\end{pmatrix}$$

$$H_X = \begin{pmatrix}
    1 & 1 & 1 & 1 & 1 & 1 & 0 & 0 & 0 \\
    0 & 0 & 0 & 1 & 1 & 1 & 1 & 1 & 1
\end{pmatrix}$$

%

The claim ``distance at least 3'' splits, by the CSS specialization, into two
SAT queries: no weight-$\leq 2$ error in the $X$-only sector and no weight-$\leq
2$ error in the $Z$-only sector. Each query uses a kernel basis for the opposite
parity-check matrix (these are also generators of the underlying classical codes
$C_1$ and $C_2$):

$$G_X = \begin{pmatrix}
    1&1&0&0&0&0&0&0&0\\
1&0&1&0&0&0&0&0&0\\
0&0&0&1&1&0&0&0&0\\
0&0&0&1&0&1&0&0&0\\
1&0&0&1&0&0&1&0&0\\
1&0&0&1&0&0&0&1&0\\
1&0&0&1&0&0&0&0&1\\
\end{pmatrix}$$

$$G_Z = \begin{pmatrix}
    1&1&1&0&0&0&0&0&0\\
0&0&0&1&1&1&0&0&0\\
0&0&0&0&0&0&1&1&1\\
\end{pmatrix}$$

%

The $X$-sector query (commutation with the six $Z$-stabilizer rows of $H_Z$,
non-membership in the row space of $H_X$, weight bound) instantiates to:

$$E_1 \oplus E_2 = 0 \wedge E_2 \oplus E_3 = 3 \wedge E_4 \oplus E_5 =0 \wedge
E_5 \oplus E_6 = 0 \wedge E_7 \oplus E_8 = 0 \wedge E_8 \oplus E_9 = 0 $$
$$\wedge$$
$$E_1 \oplus E_2 = 1 \vee E_2 \oplus E_3 = 3 \vee E_4 \oplus E_5 =1 \vee E_4
\oplus E_6 = 1 \vee E_1 \oplus E_7 = 1 \vee E_1 \oplus E_8 = 1 \vee E_1 \oplus
E_9 = 1 $$

$$|E| < 3$$
    
\end{mybox}

\subsection{Practical Verification for Larger Codes}

The Shor example fits comfortably inside Lean with the straightforward encoding.
Industrial-scale codes do not. The solver-based pipeline has three stages:
\begin{enumerate}[leftmargin=1.5em]
    \item \textbf{Formula construction:} translate the parity-check matrices and
    the distance bound into a propositional formula.
    \item \textbf{Solver dispatch:} hand the formula to CaDiCaL and obtain
    either a satisfying assignment or an LRAT UNSAT certificate.
    \item \textbf{Proof reconstruction:} re-check the certificate in Lean's
    kernel.
\end{enumerate}
Each stage breaks first for a different code size, and each demands its own
optimization. We describe the bottlenecks in the order they appear as code size
grows.

\paragraph{Stage 1: \lstinline{Matrix}-to-\lstinline{BitVec} flattening.} The
naive encoding builds the formula directly out of Mathlib's \lstinline{Matrix (Fin k) (Fin n) (ZMod 2)} type. \lstinline{Matrix} is
internally a function
$\mathrm{Fin}\,k \times \mathrm{Fin}\,n \to \mathbb{Z}_2$, so constructing the
formula requires materializing every entry as a separate Lean expression. The
memory footprint blows up before the $[[90,8,10]]$ BB code instance even reaches
the solver.

We introduce a \lstinline{BitVec}-flattened intermediate representation: a
matrix of type \lstinline{Matrix (Fin k) (Fin n) (ZMod 2)} is encoded as a
single \lstinline{BitVec (k * n)}, with vectors handled analogously.
\lstinline{BitVec} is the type Lean's program-verification stack is optimized
for and the type \lstinline{bv_decide} accepts natively; flattening once and
operating on the \lstinline{BitVec} carries us comfortably past 90 qubits. The
pipeline includes a proven equivalence between the \lstinline{BitVec} and
\lstinline{Matrix} representations, so the original quantum-theoretic distance
statement still applies.

\paragraph{Stage 2: location-based error encoding.} With the encoding fixed, the
next wall is the solver itself. A one-variable-per-qubit error encoding produces
formulas with as many independent variables as qubits, and the SAT solver scales
exponentially in this count. Empirically the 90-qubit BB instance is already at
the edge of practicality, and 144 qubits is intractable outside Lean under this
encoding.

We exploit the weight bound implicit in distance verification: an error
witnessing distance $< d$ has Hamming weight at most $d-1$, so it is uniquely
described by a list of at most $d-1$ \emph{locations} rather than by an $n$-bit
vector. We replace the per-bit variables with $(d-1)$ location variables of
width $\lceil \log_2 n \rceil$, reducing the independent-variable count from $n$
to $k\lceil \log_2 n \rceil$ where $k \leq d-1$.

The effect is concrete. For the $[[144, 12, 12]]$ BB code, $d = 12$ means $k
\leq 11$ and $\lceil \log_2 144 \rceil = 8$, so the location encoding uses $88$
variables versus $144$ --- a $56$-variable reduction. The change does not
literally divide the search space by $2^{56}$ (the constraints couple the
variables non-trivially), but in practice it brings the 144-qubit instance into
the range CaDiCaL discharges in minutes outside Lean.

A further search space reduction is performed through \textbf{symmetry breaking}. In addition to the location constraints, we introduce a constraint enforcing that each location variable be less than or equal to the succeeding variable. In the naive setting, this reduces the search space factorially with respect to the code distance. 

\paragraph{Stage 3: proof reconstruction.} With Stages 1 and 2 in place,
certificate reconstruction inside Lean's kernel has gone through smoothly
through the $[[90, 8, 10]]$ BB code. Beyond that size we expect Lean's LRAT
replay to become the next bottleneck. Lean provides an option to trust the SAT
solver's output without re-checking it; doing so would extend our scaling
envelope at the cost of admitting CaDiCaL into the trusted base, a trade-off we
view as clean enough to be worth taking when the verification target is a code
too large for kernel replay.

\subsection{Kernel Checking}

%
%

The SAT translation requires a generating set for the kernel of each
parity-check matrix. Outside Lean this is a one-line linear-algebra call; inside
Lean we need either a verified implementation of Gaussian elimination or a way
to certify that an externally-supplied basis really generates the kernel. We
take the second route: we accept the kernel as an opaque input and discharge a
Lean-side correctness check by the following corollary of rank-nullity, which
avoids any need to formalize row-reduction.

\begin{lemma}\label{lma41}
    Let $M_1, M_2$ be binary matrices with dimensions $k_1 * n$ and $k_2 * n$
    respectively. Let $r_1 + r_2 = n$. If every row of $M_1$ is orthogonal to
    every row of $M_2$, and $r_1 \le \text{rank } M_1$ and $r_2 \le \text{rank }
    M_2$, then the rowspace of $M_2$ is the kernel of $M_1$.
\end{lemma}

%

The kernel-correctness obligation splits into two sub-obligations, each handled
by its own technique below: a rank lower bound on the parity-check matrix and on
the supposed kernel (\Cref{sec:rank-check}), and mutual row-orthogonality
between them (\Cref{sec:orth-check}).

\subsubsection{Rank and Independence Checking}
\label{sec:rank-check}

%
%

\Cref{lma41} asks for a rank \emph{lower bound}, not full row-independence. The
relaxation matters in practice because many codes specify dependent generator
sets: the Toric code's full stabilizer set is dependent, and each parity-check
matrix of the $[[90, 8, 10]]$ BB code carries 4 dependent rows. A formalization
that required full independence would have to reject these specifications or
silently drop rows --- both unappealing.

To certify $\mathrm{rank}(M) \geq r$, it suffices to exhibit a row submatrix
with $r$ linearly independent rows. We discharge the linear independence of that
submatrix by another SAT query: the submatrix has full rank iff no nontrivial
$\mathbb{F}_2$-coefficient vector annihilates its rows, which is a finite
Boolean statement. The independence query consistently runs much faster than the
distance query for the same code --- circumstantial evidence that CaDiCaL is
performing something like an internal Gaussian elimination over the smaller
variable set.

\subsubsection{Orthogonality}
\label{sec:orth-check}

%
%

CSS codes additionally require the dual-containment condition $C_1^\perp
\subseteq C_2$, which on the parity-check matrices is mutual row-orthogonality.
The same condition couples the parity-check matrix to its supplied kernel basis
under \Cref{lma41}. Both checks reduce to a polynomial number of
$\mathbb{F}_2$-dot-product evaluations between rows.

We define \lstinline{bitvec_mutually_orth} on a pair of \lstinline{BitVec}s,
verified consistent with the matrix-level statement, and discharge concrete
instances by evaluating the predicate down to a Boolean. Because every building
block is decidable --- no use of the choice operator anywhere on the path ---
the goal reduces by unfolding. We use Lean's \lstinline{native_decide} tactic
rather than \lstinline{simp} for this final reduction: it compiles the decision
procedure to native code and runs orders of magnitude faster than the kernel
simplifier. The trade-off is that \lstinline{native_decide} extends the trusted
base with Lean's compiler; we accept this for orthogonality checks but not,
e.g., for the distance reduction itself.

\section{Case Studies}\label{sec:casestudies}

%

This section instantiates the pipeline of \Cref{sec:smt} on three CSS codes of
increasing size: the $[[7,1,3]]$ Steane code, the $[[23,1,7]]$ Golay code, and
the $[[90,8,10]]$ Bivariate Bicycle code. Each case study follows the same
template --- input matrices, kernel certification, distance SAT query --- so the
development is reusable: instantiating the framework on a new CSS code amounts
to supplying the parity-check matrices, their kernel bases (computed offline),
and the chosen distance bound, and renaming. We close the section by reporting
the resource envelope outside Lean for the $[[144, 12, 12]]$ BB code.

\subsection{The Steane Code}

%

The $[[7,1,3]]$ Steane code is the simplest non-trivial case study. It is a CSS
code with $H_X = H_Z$, both equal to the parity-check matrix of the $[7,4,3]$
binary Hamming code; consequently $d_X = d_Z$ and the distance argument is
carried out once.

$$H_x = H_z = \begin{pmatrix}
    1 & 0 & 0 & 1 & 0 & 1 & 1 \\ 
    0 & 1 & 0 & 1 & 1 & 0 & 1 \\
    0 & 0 & 1 & 0 & 1 & 1 & 1
\end{pmatrix}$$

%

The matrix is entered directly as a Mathlib \lstinline{Matrix} for human
readability:

\begin{lstlisting}
def steane_mat' : Matrix (Fin 3) (Fin 7) (ZMod 2) :=
  !![1, 0, 0, 1, 0, 1, 1;
    0, 1, 0, 1, 1, 0, 1;
    0, 0, 1, 0, 1, 1, 1]
\end{lstlisting}

%

All subsequent operations happen on the \lstinline{BitVec}-flattened form
(\Cref{sec:translation}). An external Python script computes the hexadecimal
flattening (\texttt{0x1d2d69} for the Steane matrix) and we discharge the
equivalence to the \lstinline{Matrix} form by Lean evaluation of the verified
\lstinline{flatten_matrix} function.

\begin{lstlisting}
def steane_mat : BitVec (3 * 7) :=
  0x1d2d69
lemma steane_mat_correct : steane_mat = flatten_matrix steane_mat' := by decide
\end{lstlisting}

%

The kernel basis is computed by the same script and entered alongside its
hexadecimal flattening.

\begin{lstlisting}
def steane_ker' : Matrix (Fin 4) (Fin 7) (ZMod 2) :=
 !![1, 1, 0, 1, 0, 0, 0;
 0, 1, 1, 0, 1, 0, 0;
 1, 0, 1, 0, 0, 1, 0;
 1, 1, 1, 0, 0, 0, 1;]
def steane_ker : BitVec (4 * 7) :=
  0x8e94b0b
lemma steane_ker_correct : steane_ker = flatten_matrix steane_ker' := by decide
\end{lstlisting}

%
%

\Cref{lma41} demands a rank lower bound on both the parity-check matrix and the
supplied kernel. For the Steane code both matrices have full row rank, so the
witness submatrix is the matrix itself and the rank bound reduces to
row-independence. We discharge independence via the SAT route of
\Cref{sec:rank-check}: rewrite to the \lstinline{BitVec} form, normalize to the
shape \lstinline{bv_decide} accepts, and let CaDiCaL verify that no nontrivial
$\mathbb{F}_2$-combination of rows is zero.

\begin{lstlisting}
lemma steane_ind : LinearIndependent (ZMod 2) steane_mat' := by
  rw [linear_indep_SAT_correct, ←steane_mat_correct, steane_mat]
  simp only [... long lemma list ...]
  bv_decide
lemma steane_ker_ind : LinearIndependent (ZMod 2) steane_ker' := by
  rw [linear_indep_SAT_correct, ←steane_ker_correct, steane_ker]
  simp only [... long lemma list ...]
  bv_decide
\end{lstlisting}

%

Linear independence then yields the rank bound used in \Cref{lma41}.

\begin{lstlisting}
lemma steane_mat_rank : 3 ≤ steane_mat'.rank := by
  apply Matrix.rank_le_of_submatrix_independent _ id (strictMono_id)
  rw [Matrix.submatrix_id_id]
  apply steane_ind

lemma steane_ker_rank : 4 ≤ steane_ker'.rank := by
  apply Matrix.rank_le_of_submatrix_independent _ id (strictMono_id)
  rw [Matrix.submatrix_id_id]
  apply steane_ker_ind
\end{lstlisting}

%

The two orthogonality conditions from \Cref{sec:orth-check} are the CSS
dual-containment ($H_X$'s rows mutually orthogonal --- here against themselves)
and the parity-check-vs-kernel orthogonality from \Cref{lma41}. Both expand to
fully-concrete $\mathbb{F}_2$ dot products with no free variables, so SAT is
unnecessary; \lstinline{simp} on the \lstinline{BitVec} unfoldings suffices.

\begin{lstlisting}
lemma steane_orth : steane_mat'.mutually_orth_rows steane_mat' := by
  rw [←mutually_orth_correct, ←steane_mat_correct, steane_mat]
  simp [bitvec_mutually_orth, mutually_orth_sat_aux,
  BitVec.row, BitVec.dot_product, dot_product_aux]

lemma steane_mat_ker_orth : steane_mat'.mutually_orth_rows steane_ker' := by
  rw [←mutually_orth_correct, ←steane_mat_correct, ←steane_ker_correct, steane_mat, steane_ker]
  simp [bitvec_mutually_orth, mutually_orth_sat_aux,
  BitVec.row, BitVec.dot_product, dot_product_aux]
\end{lstlisting}

%

The core distance lemma is the SAT query proper: every error vector under the
weight bound either fails commutation with the kernel basis (i.e.\ lies in the
row space) or is the zero error. We rewrite into the location-indexed
\lstinline{BitVec} form and dispatch \lstinline{bv_decide}.

\begin{lstlisting}
set_option maxHeartbeats 0 in
lemma steane_dist : lt_dist_sat (flatten_matrix steane_mat') (flatten_matrix steane_ker') 2 3
  := by
  rw [←steane_mat_correct, ←steane_ker_correct, steane_mat, steane_ker]
  simp only [... long lemma list ...]
  bv_decide (timeout := 999)
\end{lstlisting}

%

With rank bounds and orthogonality in place, \Cref{lma41} certifies that the
externally-supplied basis is in fact a basis of $\ker H$.

\begin{lstlisting}
lemma steane_ker_mat_correct : steane_ker'.is_ker_for steane_mat' := by
  apply Matrix.is_ker_for_of_rank_sum_mutually_orth _ _ steane_mat_rank steane_ker_rank (by norm_num) steane_mat_ker_orth
\end{lstlisting}

%

The final step assembles the CSS code from its stabilizer generators and applies
the verified SAT-translation soundness theorem to lift the BSM-level UNSAT
result to a quantum-theoretic distance bound on \lstinline{steane_css}.

\begin{lstlisting}
theorem steane_dist_3 : 3 ≤ steane_css.toBSM.distance (by norm_num) := by
  apply bitvec_sat_translation_correct steane_css steane_ker' steane_ker_mat_correct steane_ker' steane_ker_mat_correct (by norm_num) _ _
  all_goals rw [steane_css, CSS_pair.of_matrices]
  all_goals exact steane_dist
\end{lstlisting}

%
%
%
%

\subsection{Beyond Steane}

\paragraph{Golay $[[23,1,7]]$.} The $[[23,1,7]]$ Golay code is another weakly
self-dual CSS code ($H_X = H_Z$), so its Lean development is, mechanically, a
rename of the Steane case study: regenerate the kernel and hexadecimal
flattening with the Python script, swap every ``Steane'' identifier for
``Golay,'' and the proofs go through unchanged. The script and the Lean file
together are under 110 lines.

\paragraph{Bivariate Bicycle $[[90,8,10]]$.} The $[[90,8,10]]$ BB code is the
first case study with two genuinely different parity-check matrices ($H_X \neq
H_Z$). Two distance queries are needed, one for $d_X$ and one for $d_Z$. The
parity-check matrices also have $4$ dependent rows each, so the rank-bound step
from \Cref{sec:rank-check} cannot collapse to row-independence as in the Steane
case; instead, the Python preprocessor selects an independent row subset to feed
into the SAT-based independence checker. Every preprocessor output is re-checked
inside Lean, so although the preprocessing happens outside the kernel, the
resulting distance certificate remains end-to-end. The main proof file is just
under 200 lines. We autogenerate similar proofs for BB codes of size 18, 30, 36, 54, 60, 72, and 108 with the same format using a Python script. 

\paragraph{$[[144,12,12]]$ outside Lean.} The 144-qubit BB instance pushes
Lean's certificate-replay budget beyond what we can verify in the kernel. Under
the location-indexed encoding with symmetry breaking of \Cref{sec:translation}, however, an external call to the SMT solver cvc5 dispatches this goal in minutes. We report
this number to show that the SAT encoding itself scales, and we view kernel-side
replay at this size as the next concrete engineering target (cf.\
\Cref{sec:future}).

\paragraph{Template summary.} The three Lean case studies share a common
skeleton: matrix entry, \lstinline{BitVec}-flattening equivalence, kernel basis
entry, rank/independence via SAT, mutual orthogonality, distance SAT query, and
the final lift through \lstinline{bitvec_sat_translation_correct}. Adding a new
CSS code requires only its parity-check matrices and a Python preprocessor
invocation. We provide a python script that generates a Lean file verifying the distance of a CSS code given only a NumPy description of the code stabilizers, demonstrating the flexibility of this template.

\section{Discussion and Conclusion}\label{sec:Conclusion and Discussion}

%
%

We presented \textsc{Lean-QEC}, a Lean 4 framework that combines interactive
theorem proving with SAT-based computation to deliver end-to-end,
machine-checked distance certificates for practical CSS and Bivariate Bicycle
codes. The library contributes a reusable stabilizer-code stack, a verified
reduction from code distance to Boolean satisfiability, two scalability
techniques (\lstinline{BitVec} flattening and location-indexed error encoding),
and case studies covering the Steane, Golay, and a family of Bivariate Bicycle codes up to $[[90, 8, 10]]$ inside
the kernel and the $[[144,12,12]]$ instance outside it.

The result narrows the gap that has separated the testing-only and solver-only
options for QEC verification. Every step from the quantum-theoretic definition
of distance, through the binary symplectic representation, through the
\lstinline{BitVec} encoding, to the final SAT certificate is connected by a
Lean-checked theorem. Either the kernel accepts the artifact or it does not; no
informal mathematical glue sits in between.

\subsection{Methodological Takeaways}\label{sec:method}

%
%
%
%
%

The build process surfaced two design lessons that we expect to transfer to
future quantum formalizations.

\paragraph{Representation bridges are first-class objects.} The binary
symplectic representation is the obvious algebraic bridge between the
quantum-theoretic stabilizer code and the matrix-level SAT query. The
non-obvious step is that this bridge must itself be a formalized object, not an
informal correspondence, if the resulting certificate is to be end-to-end. We
further bridged the BSR to a \lstinline{BitVec} representation tuned for the
solver. Each bridge is one Lean theorem and, once written, is invisible at the
call site; without it, every downstream lemma would carry an unverified
soundness gap.

\paragraph{Verified linear algebra is the missing piece.} The pipeline still
admits one externally computed input: a kernel basis for each parity-check
matrix. We discharge its correctness via a rank/orthogonality argument, but the
kernel itself is supplied as a definition rather than computed. A verified
Gaussian-elimination implementation in Lean would let the framework take
parity-check matrices alone as input. More broadly, the absence of efficient
verified algebraic computation is the consistent friction point in this
development; closing it would benefit many Lean formalizations beyond QEC.

\subsection{Related Work}\label{sec:related}

%
%

\paragraph{Quantum verification in proof assistants.} The SQIR
development~\cite{hietala2021} and CoqQ~\cite{coqq} formalize quantum circuits
and circuit optimizations in Rocq, and Peng et
al.~\cite{peng2022formallycertifiedendtoendimplementation} verify end-to-end
correctness of algorithms such as Shor's. QHLProver~\cite{liu2019} formalizes
quantum Hoare logic. On the Lean side,
Lean-QuantumInfo~\cite{meiburg2025formalization} formalizes finite-dimensional
quantum-information primitives but does not specialize to qubit systems or QEC;
we build on the same Mathlib foundations but specialize to the qubit space and
develop the stabilizer-code stack.

\paragraph{Verification of QEC.} The 9-qubit Shor-code verification
of~\cite{feng-9-qubits} closes the trust loop on a single small code but
enumerates Pauli errors explicitly and does not scale beyond it. Other
verified-QEC efforts~\cite{huang2025efficient, wu2021qecv} sit outside ITPs and
target program-level properties rather than code-distance certification. To our
knowledge, \textsc{Lean-QEC} is the first development that combines a
kernel-checked stabilizer stack with a solver-backed distance pipeline that
scales to industrial code sizes.

\paragraph{Solver-assisted ITP.} The strategy of reducing a hard subgoal to SAT,
dispatching it to an external solver, and reconstructing the certificate inside
the kernel is well established in classical verification. Our contribution to
this idiom is the verified reduction from a quantum-theoretic distance statement
to a SAT formula, the \lstinline{BitVec}-flattened encoding that makes the
formula tractable to instantiate, and the location-indexed variable scheme that
makes the formula tractable to solve.

Concurrent work by Jain~\cite{qeclean} also explores formalization of quantum
error-correcting codes in Lean but focuses on the toric code family.

\subsection{Limitations and Future Work}\label{sec:future}

\paragraph{Solver scaling.}
The SAT pipeline scales by orders of magnitude over hand-written ITP distance
proofs --- our 144-qubit instance handles roughly $2^{40}$ candidate errors ---
but it is not unbounded. Codes substantially larger than the BB family
($\mathrm{BB}_{360}$ and beyond) push the SAT query itself out of reach.
Decomposition strategies and specialized solvers for
linear codes are natural directions for extending the envelope.

\paragraph{Broader code families.}
The development targets CSS codes, with BB as the headline family. Extending to
general non-CSS stabilizer codes is mostly a matter of dropping the $X$/$Z$
split and reusing the full BSR pipeline; subsystem codes and codes from the
broader stabilizer formalism require more substantial additions. Note that
surface and color codes are themselves CSS and should fit the existing framework
directly.

\paragraph{Verified kernel computation.}
The kernel basis for each parity-check matrix is currently produced by an
external Python script and certified inside Lean. A verified row-reduction
routine would close this remaining preprocessor step. Doing so cleanly is the
most concrete piece of follow-on engineering implied by this work.

\paragraph{Integration with quantum-circuit verification.}
The QEC layer is one component of the verified quantum-computation stack.
Connecting it to verified circuit compilers such as SQIR would yield an
integrated, end-to-end pipeline from algorithm specification to fault-tolerant
physical implementation.

\section{Acknowledgements and Code Availability}
\label{sec:acknowledgements}

We thank Junyi Liu for many useful discussions.
We thank our AI assistant Aristotle \cite{achim2025aristotle} for its contributions to the repository. YL and XW are partially supported by the U.S. National Science Foundation grant CCF-1942837 (CAREER), CCF-2330974, and a Sloan Research Fellowship. 

The library may be found at \hyperlink{https://github.com/VerifiedQC/Lean-QEC}{https://github.com/VerifiedQC/Lean-QEC}. The verified BB codes are located within the library at LeanQEC/Stabilizer/Examples/BB.

\bibliographystyle{plain}
\bibliography{references}

\end{document}